\documentclass[letterpaper, 10pt, conference]{ieeeconf}  % Comment this line out
\IEEEoverridecommandlockouts                              % This command is only
\usepackage{graphics} % for pdf, bitmapped graphics files
\usepackage{epsfig} % for postscript graphics files
\usepackage{amsmath} % assumes amsmath package installed
\usepackage{amssymb}  % assumes amsmath package installed
\usepackage{xcolor}
\usepackage{mathtools}
\usepackage{cite}
\usepackage{hyperref}

\title{\LARGE \bf {Propagation Stability Concepts for Network Synchronization Processes}} 

\author{Sandip Roy, Subir Sarker, and Mengran Xue% <-this % stops a space
\thanks{The first two authors are with School of Electrical Engineering and Computer Science,
        Washington State University, Pullman, WA 99164, USA. The third author
        is with Raytheon BBN Technologies.
        {\tt \small Correspondence: sandip@wsu.edu}}%
}

\newtheorem{definition}{Definition}

\newtheorem{theorem}{Theorem}

\newtheorem{corollary}{Corollary}

\begin{document}
\maketitle
%\thispagestyle{empty}
%\pagestyle{empty}

%%%%%%%%%%%%%%%%%%%%%%%%%%%%%%%%%%%%%%%%%%%%%%%%%%%%%%%%%%%%%%%%%%%%%%%%%%%%%%%%
\begin{abstract}
A notion of disturbance propagation stability is defined for dynamical network processes, in terms of decrescence of an input-output energy metric along cutsets away from the disturbance source.  A characterization of the disturbance propagation notion is developed for a canonical model for synchronization of linearly-coupled homogeneous subsystems.  Specifically, propagation stability is equivalenced with the frequency response of a certain local closed-loop model, which is defined from the subsystem model and local network connections, being sub-unity gain.  For the case where the subsystem is single-input single-output (SISO), a further simplification in terms of the subsystem's open loop Nyquist plot is obtained.  An extension of the disturbance propagation stability concept toward imperviousness of subnetworks to disturbances is briefly developed, and an example focused on networks with planar subsystems is considered. 
\end{abstract}

\section{Introduction}

The synchronization of coupled systems is common in nature and in the engineered world.  For this reason, there has been a substantial cross-disciplinary research effort to define and characterize synchronization phenomena in networks of coupled systems \cite{pecora,chua,murray}.  These studies define synchronization in terms of the internal asymptotic stability of a manifold, on which each coupled system has an identical state or output.  However, synchronization in a practical sense requires not only that the coupled systems come to a common state/output, but also that this equilibrium is impervious to external disturbances.  Based on this recognition, a body of recent work has considered the disturbance responses of network synchronization processes \cite{gchen,saberi,tegling}.  Broadly, these efforts model disturbances as impinging on all or a subset of subsystems within a network, and evaluate their potential impacts on network-wide synchrony according to a performance metric (typically, a $H_2$ or $H_\infty$ gain).  The studies in this direction have largely focused on evaluating the metrics from a graph-theoretic perspective, and designing controllers to bound the performance metrics within a threshold.

Engineers working with large-scale built networks often do not think about coordination in terms of either internal stability characteristics or global (network-wide) disturbance response metrics.  Rather, they are concerned with the extent of propagation of local disturbances, whether arising from exogenous inputs or state deviations.  For instance, bulk power grid operators often distinguish well-damped networks where oscillatory disturbance responses remain localized from poorly-damped ones where such disturbances have network-wide impact \cite{gridforcedoscillation}.  Similar assessments of network performance in terms of disturbance propagation are of interest in disciplines ranging from air traffic management to infectious-disease epidemiology and cyber-security \cite{propstability}.  

The spatial propagation of disturbances among interconnected systems in cascade or line-topology configurations has been extensively researched under the heading of string stability \cite{string1,string2}.  The string stability concept has also been extended to general directional networks in the {\em mesh stability} literature \cite{mesh}. The recent study \cite{studli} has recognized the need for disturbance-propagation stability notions for general (bi-directionally connected) dynamical networks, and therefore has proposed a definition for network stability in terms of boundedness of input-to-state or input-to-output gains which parallels the basic string stability definition. It has also introduced an alternate network stability definition for tree-like graphs which captures monotonic decresence of the propagative response away from the disturbance source; this definition is a generalization of the {\em strict string stability} concept \cite{string2}.  Other recent studies have also examined stability notions for linear and nonlinear network processes, which are concerned with disturbance propagation (e.g., \cite{mirabilio}). Additionally, a number of recent studies have characterized local-input-to-output properties of dynamical networks (e.g., gains, zeros), without however explicitly considering propagation \cite{porco,koorehdavoudi}. However, definitions for propagation stability are not yet mature, and the formal analysis of propagation is incomplete even for networks of coupled linear systems. 

%Beyond these initial efforts, however, disturbance propagation metrics and propagation stability concepts have not yet been fully defined for general networks, and analyses have not been completed for networks made up of coupled linear systems.  

The purpose of this article is to define and characterize a notion of disturbance propagation stability in the context of a canonical network model for synchronization of coupled homogeneous linear subsystems.  The main contributions of the study are two-fold:

1) A general definition is introduced for (strict) propagation stability, based on decrescence of response norms across cutsets in the network's digraph, or equivalently along paths away from the disturbance source.

2) Propagation stability is characterized in terms of the closed-loop frequency responses of the subsystem model with a proportional feedback controller applied.  For the case of single-input single-output subsystems, these conditions are further simplified to conditions on the subsystem's open-loop frequency response and/or transfer function.

The analysis of propagation stability depends on a new, local characterization of synchronization models.  This characterization is markedly different from the spectral decomposition that has been exhaustively used to understand both internal and global external stability \cite{pecora,chua}, and provides a further assessment/categorization of synchronization processes.

The remainder of the article is organized as follows.  The network synchronization model is described in Section II, and disturbance propagation stability notions are defined in Section III.  The main characterizations of propagation stability are presented in Section IV.  Finally, in Section V, an example is given which focuses on the impact of damping on propagation stability when the subsytems are planar devices.

\section{Model}

A network with $N$ identical, interconnected devices or nodes or {\em subsystems}, labeled $1,\hdots, N$, is considered.  Each subsystem $i \in 1, \hdots, N$ has a state ${\bf x}_i \in R^n$ and output ${\bf y}_i \in R^m$ which are governed by the following linear or linearized state-space equations:
\begin{eqnarray}
& & \dot{\bf x}_i = A {\bf x}_i +B \left( \alpha \sum_{j \neq i} g_{ij}({\bf y}_j - {\bf y}_i) +\gamma_i {\bf w}_i \right) \label{eq:main}\\
& & {\bf y}_i=C {\bf x}_i .  \nonumber
\end{eqnarray}
Here, $A$, $B$, and $C$ are a subsystem's state, input, and output matrices, respectively; the scalars $g_{ij} \ge 0$ are coupling weights; the vector ${\bf w}_i \in R^m$ represents an external disturbance input at subsystem $i$; and $\alpha$ is a global coupling-strength parameter which allows tuning of the network connectivity (see \cite{pecora}).  Our focus here is on an external disturbance impinging on a single source node $s \in 1,\hdots, N$, which is modeled by setting $\gamma_s=1$ and $\gamma_i=0$ 
for $i \neq s$.  

The model (\ref{eq:main}) is a standard representation for the small-signal dynamics of synchronizing coupled oscillators \cite{pecora,chua}, with two distinctions.  First, a disturbance is applied at a single subsystem, to allow for analysis of propagative impacts.  Second, the subsystems are modeled as being interconnected through commensurately-dimensioned inputs and outputs, rather than through an explicit inner-coupling term or alternately through a designable protocol.  This format is used to stress that the network is made up of input-output devices with fixed connections, but the formulation encompasses the scenarios with an inner coupling or a designable protocol.  

Analyses of (\ref{eq:main}) often are phrased in terms of a graph that represents the network interconnections.  For our development, a weighted digraph $\Gamma$ is defined with $N$ vertices corresponding to the $N$ subsystems.  A directed edge is drawn from vertex $j$ to vertex $i$ if $g_{ij}>0$, reflecting a direct influence of the output of subsystem $j$ on the state evolution of subsystem $i$.  The edge is assigned a weight of $g_{ij}$.
We use the notation ${\cal V}$ for the set of vertices, and ${\cal E}$ for the set of edges.
Additionally, it is convenient to define an (asymmetric) Laplacian matrix $L=[l_{ij}] \in R^{N \times N}$. Each off-diagonal entry $l_{ij}$ is given by $-g_{ij}$, while the diagonal entries are selected so that each row sums to $0$ (i.e. $l_{ii}=\sum_{j \neq i} g_{ij}$).  

For the model (\ref{eq:main}), the synchronization manifold where the states ${\bf x}_1,\hdots, {\bf x}_n$ are identical is known to be asympotically stable under broad conditions on the network graph, the subsystem model, and the coupling-strength parameter.  
More precisely, stability can be related to Hurwitz stability of the $N$ complex matrices $A+\lambda_i BC$, where $0=\lambda_1,\lambda_2, \hdots, \lambda_N$ are the eigenvalues of the Laplacian matrix $L$.  From this analysis, stability can be distilled to a simple test on the Laplacian matrix's spectrum via the master stability function construct, see e.g. \cite{pecora} for details.  A number of other characteristics of (\ref{eq:main}), including global disturbance stability and controllability via external stimulation, can also be related to the matrices $A+\lambda_iBC$ \cite{gchen,controllability}.    

 From here on we refer to the model (\ref{eq:main})  as the network synchronization model. The model is approximative of a number of synchronization phenomena, including the swing dynamics of the bulk power grid, multi-vehicle formation flight, and the nonlinear dynamics of electrical oscillator networks.

\section{Propagation Stability Definition}

A notion of propagation stability is defined based on the spatial patterns of output energies (squared two norms) at network subsystems over a time interval $[0,T]$, when an exogenous disturbance input $w_s(t)$ is applied at a single node.  In our development, the disturbance is assumed to satisfy the Dirichlet conditions (absolute integrability over any period, finite number of discontinuities and minima/maxima, bounded over any interval), but otherwise may be arbitrary. In defining stability, the squared two-norm metric $E_i(T)=\int_{t=0}^T {\bf y}_i^T (t) {\bf y}_i(t) \, dt$ is considered for each network subsystem $i$. Conceptually, the network can be viewed as propagation stable, if these energies are attenuated away from the disturbance source with respect to the network graph.  However, since the network's graph in general has a spatially inhomogeneous structure, defining attenuation requires some care.

%For the network synchronization model, propagative responses due to either: 1) local exogenous inputs or 2) localized perturbations of the state away from the synchronization manifold are of interest.  This study is focused on characterizing propagative responses due to an exogenous input.  Precisely, the zero-state response of the network model to a disturbance $w_s(t)$ at a single source location (Equation \ref{eq:main}) is tracked after a time $t=0$.   
%The disturbance is assumed to satisfy the Dirichlet conditions (absolute integrability over any period, finite number of discontinuities and minima/maxima, bounded over any interval), but otherwise may be arbitrary.  We are concerned with the spatial pattern in the response energies (squared two norms) of the outputs at network subsystems over an interval $[0,T]$.  Formally, the metric $E_i(T)=\int_{t=0}^T {\bf y}_i^T (t) {\bf y}_i(t) \, dt$ is considered for each network subsystem $i$. 

\begin{figure}[!htb] 
\centering
\includegraphics[width=8cm]{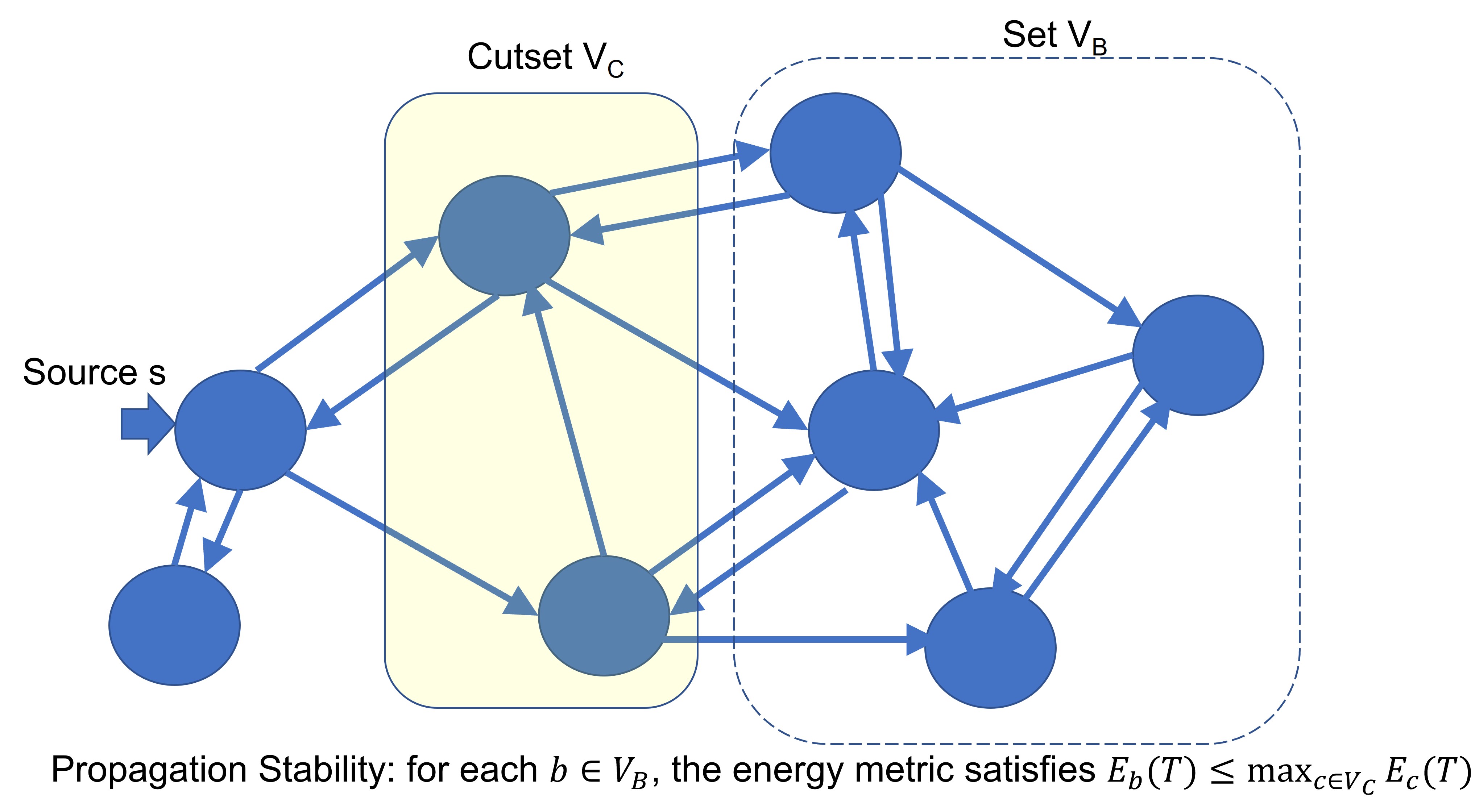}
\caption{Illustration of the propagation stability concept.} \label{fig:1}
\end{figure}

One natural way to assess disturbance propagation is to consider the energy metric $E_i(T)$ for vertex-cutsets in the network's graph (i.e., sets of vertices whose removal partition the graph). If the metric value for at least one cutset vertex is larger than the metric values for vertices that   are separated from the source by the cutset, then the response can be viewed as being attenuated away from the source (see Figure \ref{fig:1}).  To formalize this notion, let us consider a set of vertices ${\cal V}_C \in {\cal V}$.  We refer
to ${\cal V}_C$ as a separating cutset for the source $s$, if the remaining vertices 
${\cal V} \setminus {\cal V}_C$ can be partitioned into two subsets ${\cal V}_1$ and ${\cal V}_B$ such that: 1) there are no edges from vertices in ${\cal V}_1$ to vertices in ${\cal V}_B$ and 2) ${\cal V}_A={\cal V}_1 \cup {\cal V}_C$ 
contains the source vertex $s$.  Using this notation, the following definition for propagation stability is proposed:

\begin{definition}
The network synchronization model is {\bf propagation stable} if the following two conditions hold. 
\begin{enumerate}
    \item The synchronization manifold is asymptotically stable in the sense of Lyapunov.
    \item For every source location $s$, disturbance signal ${\bf w}_s(t)$, separating cutset ${\cal V}_C$ for $s$, time horizon $T>0$, and vertex $b \in {\cal V}_B$, the following majorization holds: $E_b(T) \le \max_{c \in {\cal V}_C} E_c(T)$.
\end{enumerate}
    
\end{definition}

The definition captures that the output signal energy for a subsystem associated with a graph cutset upper bounds the signal energy for all subsystems beyond the cutset, and hence the output signal energy is attenuated at cutsets away from the source.  An immediate consequence of the definition is that the maximum response energy among subsystems at a distance $r$ from the source location on the network graph, say $E(r)$, is a non-increasing function of $r$.

With some further thought, one also sees that propagation stability can be equivalently expressed in terms of paths in the network graph.  In particular, the network model is propagation stable if and only if there is at least one path in the network graph from the disturbance source to each other vertex such that the response energy is non-increasing along the path.  From this equivalent form, it is evident that propagation stability is a generalization of the classical strict string-stability definition and a recent definition for (strict) network stability in tree-like networks \cite{string2,studli}, in that decrescence or attenuation along paths away from the disturbance source is enforced.  We stress that our definition only requires attenuation along one path for each pair of vertices, so as to encompass the varying attenuation patterns that may arise in inhomogeneously-structured networks.

%The cutset-based attenuation notion used in the definition is a natural generalization to the path-based criteria for input-output stability developed for tree-like networks in \cite{???}.

%!!!Show that the notion is distinct from both internal stability and input-output stability notions for synchronization/consensus models!!!

In some circumstances, the propagation stability definition may be too rigid to capture attenuative dynamics in a network.  In particular, it is possible that small regions in a network may be susceptible to disturbance amplification, but a substantial portion of the network nevertheless attenuates disturbances.  This partial notion of propagation attenuation is captured in the following definition:

\begin{definition}
Consider a set ${\cal V}_{D}$ of vertices in the network graph (respectively, subsystems in the network model).  Assume that the induced subgraph of $\Gamma$ defined by ${\cal V}_{D}$ is strongly connected.  The subnetwork defined by ${\cal V}_{D}$ is said to be {\bf propagation impervious} if the following two conditions hold. 
\begin{enumerate}
    \item The synchronization manifold for the network synchronization model is asymptotically stable in the sense of Lyapunov.
    \item For every source location $s$, disturbance signal ${\bf w}_s(t)$, separating cutset ${\cal V}_C$ for $s$ contained within ${\cal V}_D$, time horizon $T>0$, and vertex $b \in {\cal V}_B \cap {\cal V}_D$, the following majorization holds: $E_b \le \max_{c \in {\cal V}_C} E_c$.
\end{enumerate}
    
\end{definition}

The definition asserts that disturbances that enter the propagation impervious subnetwork (whether from an outside or an inside source) then exhibit a spatial attenuation within that subnetwork.  The definition for propagation imperviousness also aligns with concepts in the strict string stability literature \cite{studli}, which allow for amplification in a radius around the disturbance source provided that attenuation is guaranteed elsewhere.

\section{Propagation Stability Analysis}

Propagation stability is concerned with the ability of network subsystems to attenuate impinging disturbances, such that a disturbance response falls off spatially in the network regardless of the source.  Thus, one might expect propagation stability to be related to {\em local} characteristics of the network synchronization model, specifically the structure of each subsystem and its interconnections with its neighbors.  In the following development, we present several conditions for propagation stability and imperviousness, which are phrased in terms of local (subsystem-level) frequency response characteristics.  Formally, these conditions are expressed in terms of the following Laplace-domain {\em local feedback transfer matrices} $H_i(s)$, defined for each subsystem $i \in 1,\hdots, N$: 
\begin{small}
\begin{equation}
    H_i(s)=(\alpha \sum_{j \in {\cal N}(i)} g_{ij}) C(sI-A + \alpha \sum_{j \in {\cal N}(i)} g_{ij}BC)^{-1} B.
\end{equation}
\end{small}
\noindent Here, the notation ${\cal N}_i$ refers to the set of {\em incoming neighbors} of vertex $i$ in the graph, i.e. the set containing vertices $j$ such that $g_{ij}>0$.

The first main result of our development is a sufficient condition for propagation stability:

\begin{theorem}
The network synchronization is propagation stable if: 
1) the synchronization manifold is asympotically stable in the sense of Lyapunov and 2) 
$\sup_{\omega} \sigma_{max}(H_i(j\omega)) \le 1 $
for all $i=1,\hdots,N$. 
%Further, if the network graph has any vertex that has an incoming edge from only one other vertex, the condition is necessary and sufficient.
\end{theorem}

\begin{proof}
The synchronization manifold is asymptotically stable by assumption, hence the first criterion for propagation stability is met. 

To verify that the second criterion for propagation stability is met, first consider the zero-state output ${\bf y}_i(t)$ over the interval $[0,T]$ for any subsystem $i$ other than the source $s$.  We 
notice that, if ${\bf y}_j(t)$ for $j \in {\cal N}(i)$  is known over the interval $[0,T]$, this is sufficient to compute ${\bf y}_i(t)$ over the interval.  Specifically, ${\bf y}_i(t)$ can be computed over $[0,T]$ by first solving the
following equation for $\overline{\bf y}_i(t)$:
\begin{eqnarray}
& & \dot{\overline{\bf x}}_i = A \overline{\bf x}_i +B \left( \alpha \sum_{j \neq i}  g_{ij}(\overline{\bf y}_j - \overline{\bf y}_i) +\gamma_i {\bf w}_i \right) \label{eq:main2}\\
& & \overline{\bf y}_i=C \overline{\bf x}_i   \nonumber
\end{eqnarray}
where (without any loss of information) $\overline{\bf y}_j(t)$ for $j \in {\cal N}(i)$ is set equal to the disturbance response ${\bf y}_j(t)$ (for Equation \ref{eq:main}) over the interval $[0,T]$ and is set to zero for $t \ge T$, and the initial conditions are assumed to be zero.  The disturbance response ${\bf y}_i(t)$ is then equal to $\overline{\bf y}_i(t)$ over the interval $[0,T]$,

Transforming (\ref{eq:main2}) to the Laplace domain and solving for $\overline{Y}_i(s)$ yields:
\begin{equation}
    \overline{Y}_i(s)=H_i(s)\sum_{j \in {\cal N}(i)} \frac{g_{ij}}{\sum_{j \in {\cal N}(i)}g_{ij}} \overline{Y}_j(s), \label{eq:maintf}
\end{equation}
where the notation $\overline{Y}_j(s)$ is used for the Laplace transform of $\overline{\bf y}_j(t)$.  Thus, 
$\overline{\bf y}_i(t)$ can be found by filtering
the signal $\overline{\bf z}_i(t)=\sum_{j \in {\cal N}(i)} \frac{g_{ij}}{\sum_{j \in {\cal N}(i)}g_{ij}} \overline{\bf y}_j(t)$ with the transfer function $H_i(s)$.  From this relationship, the energy $E_i(T)$ can be bounded in terms of $E_j(T)$, $j \in {\cal N}(i)$, as follows. First:
\begin{eqnarray}
    E_i & = & \int_{t=0}^T {\bf y}_i^T (t) {\bf y}_i(t) \, dt=\int_{t=0}^T \overline{\bf y}_i^T (t) \overline{\bf y}_i(t) \, dt \\
    & \le & \int_{t=0}^\infty \overline{\bf y}_i^T (t) \overline{\bf y}_i(t) \, dt . \nonumber
\end{eqnarray}
Then, from Parseval's theorem, it follows that:
\begin{equation}
    \int_{t=0}^\infty \overline{\bf y}_i^T (t) \overline{\bf y}_i(t) \, dt =
     \frac{1}{2\pi}\int_{\omega=-\infty}^\infty \overline{Y}_i^T (j\omega) \overline{Y}_i(j\omega) \, d\omega
\end{equation}
From the relationship between $\overline{Y}_i(s)$ and
$\overline{Z}_i(s)$, this can be further characterized as:
\begin{small}
\begin{equation}
    \int_{t=0}^\infty \overline{\bf y}_i^T (t) \overline{\bf y}_i(t) \, dt =
     \frac{1}{2\pi}\int_{\omega=-\infty}^\infty \overline{Z}_i^T (j\omega) H_i^T(j \omega) H_i(j\omega)\overline{Z}_i(j\omega) \, d\omega
\end{equation}
\end{small}
\noindent However, since the maximum singular value of $H_i(j\omega)$ is assumed to be less than or equal to $1$, it is immediate that
\begin{eqnarray}
    \int_{t=0}^\infty \overline{\bf y}_i^T (t) \overline{\bf y}_i(t) \, dt & \le &
     \frac{1}{2\pi}\int_{\omega=-\infty}^\infty \overline{Z}_i^T (j\omega) \overline{Z}_i(j\omega) \, d\omega \\
     & = & \int_{t=0}^{\infty} \overline{\bf z}_i^T \overline{\bf z}_i \, dt, \nonumber
\end{eqnarray}
where the final equality is based on another application of Parseval's theorem.  Then, noting
that the signal ${\bf z}_i(t)$ is non-zero only on $[0,T]$ and exploiting convexity, it follows
that:
\begin{align}
\int_{t=0}^{\infty} \overline{\bf z}_i^T \overline{\bf z}_i \, dt  &= \int_{t=0}^{T} \overline{\bf z}_i^T \overline{\bf z}_i \, dt \\
&\le 
\sum_{j \in {\cal N}(i)} \frac{g_{ij}}{\sum_{j \in {\cal N}(i)} g_{ij}}\int_{t=0}^T \overline{\bf y}_j^T
\overline{\bf y}_j \, dt \\
&\le \max_{j \in {\cal N}(i)} \int_{t=0}^T 
\overline{\bf y}_j^T
\overline{\bf y}_j dt =\max_{j\in {\cal N}(i)} E_j, \label{eq:strictness}
\end{align}
where the last equivalence depends on recognizing that  $\overline{\bf y}_j ={\bf y}_j$ over the interval $[0,T]$,  Combining these
inequalities and equivalences, we get:
\begin{equation}
    E_i(T) \le \max_{j \in {\cal N}(i)} E_j(T). \label{eq:keyproof}
\end{equation}
Here, Equation (\ref{eq:keyproof}) is valid for any
subsystem $i$ other than the source $s$, and for any
interval $T$. Also, the inequality is strict unless all neighbors in $j \in {\cal N}(i)$ have identical $E_j(T)$, since the second inequality in (\ref{eq:strictness}) is strict except in this circumstance.

The second criterion for propagation stability can be proved through iterative application of Equation (\ref{eq:keyproof}). To develop the proof,
consider any cutset ${\cal V}_C$ and vertex $b \in {\cal V}_B$.  Define a set ${\cal R}$ which initially contains only the vertex $b$.  Then, add to ${\cal R}$ all incoming neighbors $i$ of $b$ which satisfy $E_i(T) \ge E_b(T)$; from Equation (\ref{eq:keyproof}), there exists at least one such vertex.  If the new set ${\cal R}$ includes a vertex in ${\cal V}_C$, then the majorization in the theorem statement is proven.  Otherwise, ${\cal R}$ can again be augmented to include additional vertices $i$ which 
satisfy $E_i(T) \ge E_b(T)$; from (\ref{eq:keyproof}) and the strictness argument presented thereafter, at least one such vertex must exist.  This process can be iterated until at least one vertex in ${\cal V}_C$ is included.  Therefore, $E_i(T) \ge E_b(T)$ for some $i \in {\cal V}_C$, or equivalently $E_b(t) \le max_{i \in {\cal V}_C} E_i$.  Since this inequality holds for all source locations $s$, time horizons $T$, separating cutsets ${\cal V}_C$, and vertices $b \in {\cal V}_B$, the theorem is verified.  

\end{proof}

Theorem 1 is tight in a certain sense, for the case that the subsystem model is single-input single-output (SISO).  Specifically, for this case, consider that the condition of the theorem is not met, i.e. $\max_{\omega} \sigma_{max}(H_i(j\omega))= 
\max_{\omega} |H_i(j \omega)|>1$ for some vertex $i$.  Then for some network graphs, the network synchronization model will not be propagation stable.  To see why, consider a network graph for which the vertex $i$ has only a single incoming edge, say from vertex $j$.  It is easy to check in this case that
$Y_i(s)=H_i(s) Y_j(s)$, i.e. the disturbance response at vertex $i$ is the filtration of the disturbance response at vertex $j$ by $H_i(s)$.  Then consider a single tone disturbance $w_s(t)$ (where $s \neq i$) at a frequency $\omega$ such that $|H_i(j\omega)|>1$.   For this disturbance, the response at each subsystem is also asymptotically a single-tone frequency at the same frequency $\omega$. Since $|H_i(j\omega)|>1$, it thus follows that the amplitude of the sinusoidal response at vertex $i$ is larger than that at vertex $j$.  Thus, for a sufficiently long time horizon $T$, $E_i (T) > E_j (T)$.  Therefore, since vertex $j$ is a cutset that separates the source vertex $s$ from vertex $i$, the network model is not propagation stable.  
If the subsystem model is multi-input multi-output (MIMO), then characterization of response amplitudes for a single tone input is more complex, because the response vector at a subsystem may not coincide with the maximum amplification direction.  However, the network certainly may be susceptible to
disturbance amplification if $|H_i(j\omega)|>1$.

Conditions for propagation imperviousness in a region of the network graph can also be developed using a parallel argument to the proof of Theorem 1.  Here is the result:
\begin{theorem}
Consider a network synchronization model.  For this model, consider a set ${\cal V}_{D}$ of the vertices in the network graph (respectively, subsystems in the network model), such that the induced subgraph defined by ${\cal V}_{D}$ is strongly connected.  The subnetwork defined by ${\cal V}_{D}$ is  propagation impervious if the following two conditions hold: 
1) the synchronization manifold is asympotically stable in the sense of Lyapunov and 2) 
$\sup_{\omega} \sigma_{max}(H_i(j\omega)) \le 1 $
for all $i\in {\cal V}_D$. 
\end{theorem}
Since the proof is analogous to that of Theorem 1, it is omitted. \newline

%The condition on the frequency responses of the local feedback transfer functions for propagation stability and imperviousness can be sometimes be relaxed based on structural properties of the network graph.  One approach for
%obtaining more relaxed conditions is illustrated in Figure \ref{fig:???}.  
%The idea is to partition a set of vertices ${\cal V}_D$ in the network graph into multiple subsets, which each are cutsets of the graph.  !!!Still working on this, I believe there's a way to allow for a gain larger that 1, but it is tricky!!!  !!!As I work through examples, it's clear that I need to think through this!!!  !!!Working more on this, it seems that examples without the strict requirement on H(j\omega) have rather complicated spatial responses -- largely
%will not satisfy the propagation stability criteria.

Since propagation stability depends critically on the gains of the local
transfer matrices $H_i(s)$, it is useful to further interpret these matrices from a system-theoretic standpoint.  Of interest, the transfer matrices can be interpreted as the closed-loop reference-signal-to-output transfer function
when a certain static feedback controller is applied to the subsystem model.  Specifically, the
local transfer matrix $H_i(s)$ is the transfer function from $r(t)$ to $y(t)$ 
of the following closed-loop system model:
\begin{eqnarray}
\dot{\bf x}=A{\bf x}+B {\bf u} \\
{\bf y}=C {\bf x} \nonumber \\
{\bf u}=(\alpha \sum_{j \ne i} g_{ij} ) ({\bf r}- {\bf y}).
\end{eqnarray}
The local transfer matrices are thus seen to capture the closed-loop dynamics of
the subsystem model, when an identical proportional feedback controller 
with gain $k_i=\alpha \sum_{j \ne i} g_{ij}$ applied at each channel.
This feedback control interpretation is useful for characterizing 
propagation stability in terms of only the subsystem model, as we will do 
for the SISO case in the following section.

%!!!Develop the connection with the classical stability and H2/H_=inf story!!!

\subsection{SISO Case: Subsystem-Based Characterization}

The interpretation of the local transfer matrix $H_i(s)$ as a closed-loop model allows easy development of conditions for propagation stability phrased in terms of either the frequency response or the transfer function of the subsystem model, when the subsystem is SISO.  To develop this
analysis, let us denote the transfer function for the subsystem model as $T(s)=C(sI-A)^{-1}B$.  Then, for the SISO case, the local transfer matrix can be written as 
\begin{equation}
    H_i(s)=\frac{k_iT(s)}{1+k_iT(s)},
\end{equation}
where $k_i =\alpha \sum_{j \ne i} g_{ij}$.

We now develop a check for propagation stability in terms of the subsystem frequency response $T(j\omega)$.  To do so, recall from Theorem $1$ that propagation stability requires $|H_i(j\omega)| \le 1$ for all $\omega$ and $i=1,\hdots, N$, which we call the {\em local requirement}, in addition to stability of the synchronization manifold.  Through substitution, the local
requirement can equivalently be written as 
$|\frac{k_i T(j\omega)}{1+k_i T(j\omega)}| \le 1$ for 
all $i=1, \hdots, N$.  Noting $T(j\omega)$ is a complex number for
each $\omega$, say $a(\omega)+jb(\omega)$, the expression for the local requirement 
can be written as $|\frac{k_i (a(\omega)+jb(\omega))}{1+k_i (a(\omega)+jb(\omega))}| \le 1$.  From this expression, it is apparent that the local requirement is met if and only if $|k_i a(\omega)|\le |1+k_i a(\omega) |$ for all $\omega$ and $i=1,\hdots, N$.  However, this inequality holds true if and only 
if $a(\omega)\ge\frac{1}{-2k_i}$ for all $\omega$ and all $i=1,\hdots, N$.  Thus, the local requirement is seen to be met if
and only if the following condition holds:
\begin{equation}
  Re(T(j \omega)) \ge  -\frac{1}{2\alpha \max_i \sum_{j\ne i} g_{ij}}.
\end{equation}
Thus, we see that the local requirement for propagation stability is met if and only if the frequency response $T(j\omega)$ of the local subsystem model lies entirely to the right of
$-\frac{1}{2\alpha \max_i \sum_{j\ne i} g_{ij}}$ in the complex plane.

%This result is formalized next as a necessary condition for propagation stability:

%\begin{theorem}
%Consider a network synchronization model with a SISO subsystem model, which has frequency response $T(j\omega)$.  The network model can be propagation stable only if $Re(T(j \omega)) \ge  -\frac{1}{2\alpha \max_i \sum_{j\ne i} g_{ij}}$ for all $\omega$.
%\end{theorem}

For the SISO subsystem case, the characterization of the local requirement provides a means to verify propagation stability entirely in terms of the frequency response of the subsystem model.  Specifically, both the local requirement and the standard condition for asymptotic stability of the synchronization manifold can be determined from the {\em Nyquist plot} of the subsystem model.  For the local requirement to be met, the Nyquist plot must lie entirely to the right of the vertical line  with intercept
$-\frac{1}{2\alpha \max_i \sum_{j\ne i} g_{ij}}$.  Meanwhile, the characterization of asymptotic stability of the synchronization manifold in terms of the Nyquist plot is well known \cite{chennyquist}. Briefly, 
asymptotic stability of the manifold can be equivalenced with simultaneous Hurwitz stability of matrices  $A+ \lambda_i BC$, $i=2,\hdots, N$, where the $\lambda_i$ ($i=2,\hdots, N$) are the non-zero eigenvalues of the Laplacian matrix $L$.  
For the SISO subsystem case, the matrices $A+\lambda_i BC$ are
the state matrices of the closed-loop system, when a proportional feedback controller with gain $\lambda_i$ is applied to the subsystem.  Thus, it is seen that Hurwitz stability of each matrix $A+\lambda_i BC$ can be determined through the standard Nyquist criterion, i.e. by comparing encirclements of the point $\frac{-1}{\lambda_i}$ by the Nyquist plot with the number of open-loop right-half-plane poles.  Since this stability analysis for the synchronization manifold has been developed extensively in previous work, details are omitted.

The propagation stability analysis for the SISO subsystem case gives insight into the role of the coupling constant $\alpha$.  As the coupling is weakened, the permissible region for the Nyquist plot such that the local requirement holds is widened, including more of the left half of the complex plane.  In fact, by decreasing the coupling, the local requirement can be met for any subsystem model unless the Nyquist plot diverges in the left-half-plane (which corresponds to subsystem models with repeated poles on the $j\omega$ axis).  This ability to meet the local requirement by reducing the coupling is logical, since reduced couplings should attenuate propagation of a disturbance through the network.  However, scaling down the coupling may also influence the stability of the synchronization manifold.  Analyses using the master stability function have shown the intermediate couplings are needed for stability of the synchronization manifold, for many subsystem models (e.g., several chaotic oscillators like the Rossler oscillator).  These characterizations, which can also be verified through the Nyquist analyses presented above, indicate that the local requirement and manifold stability may conflict for some subsystem models: a small gain may be needed for the local requirement to hold, while a larger gain is needed for stability of the synchronization manifold.

The frequency-domain analysis developed above immediately yields conditions for propagation stability in terms of the subsystem transfer function.  First, general characterizations of propagation stability can be obtained in the cases where the
subsystem model is either strictly unstable (has open right-half-plane poles) or strictly stable:

\begin{corollary}
Consider the network synchronization model, and assume that the subsystem model is SISO.
\begin{itemize}
    \item If the subsystem model has an open right-half-plane pole, then the network synchronization model is not propagation stable for any network graph $\Gamma$ and  coupling constant $\alpha$.
    \item If the subsystem model is strictly stable (has all poles strictly in the open left-half-plane), then there exists a positive constant $\overline{\alpha}$ such that the network synchronization model is propagation stable for $\alpha \le \overline{\alpha}$.
\end{itemize}
\end{corollary}

The result follows immediately from the Nyquist-plot-based characterization of propagation stability (both internal stability and the local requirement), hence details are omitted.

Keener characterizations of propagation stability can be developed by considering what transfer functions meet the local requirement for propagation stability.  First, we note that the local requirement is met regardless of the coupling strength and the network graph, if the Nyquist plot lies in the closed right half of the complex plane.  The restriction of the Nyquist plot in the closed right half plane is a standard notion in systems theory -- positive realness -- which holds true for a system if and only if that system is passive.  Thus, we can easily obtain a general characterization for propagation stability when the subsystem model is passive:

\begin{corollary}
Consider the network synchronization model, and assume that the subsystem model is SISO and passive.  Also, assume that the synchronization manifold is asymptotically stable.  Then 
the network synchronization model is necessarily propagation stable.
\end{corollary}

In the case where the subsystem model is not passive, propagation stability depends on the  network graph and the coupling strength, as well as the extent to which the subsystem's Nyquist plot encroaches on the left half of the complex plane. Specifically, the subsystem's frequency response is allowed to have a real part which is lower bounded by 
$-\frac{1}{2\alpha \max_i \sum_{j\ne i} g_{ij}}$ rather than by zero: thus, the graph vertex with largest total incoming edge weight along with the coupling strength decide the region where the Nyquist plot may lie.  The requirement on the subsystem model in this case may be viewed as an "almost passivity" requirement, i.e. a requirement that the system could be made passive using a direct feed-through term subject to a bound.  The requirement is formalized in the following theorem, which summarizes the frequency-response-based characterization of propagation stability:

\begin{theorem}
Consider the network synchronization model, and assume that the subsystem model is SISO.  Also, assume that the synchronization manifold is asymptotically stable.  If $Re(T(j\omega)) \ge -\frac{1}{2\alpha \max_i \sum_{j\ne i} g_{ij}}$ for all $\omega$, then
the network synchronization model is propagation stable.
\end{theorem}

%Equivalently, this condition can be written as follows:
%\begin{equation}
%    H_i(j\omega)=Ce^{j\phi} \label{eq:goo}
%\end{equation}
%for some $C\in[0,1]$ and $\phi \in [0,2\pi]$, for each $\omega$.  Substituting 
%$H_i(j\omega)=\frac{k_i T(j\omega)}{1+k_iT(j\omega)}$ and rearranging yields
%the following requirement for propagation stability:
%\begin{equation}
%    k_i T(j\omega)=\frac{Ce^{j\phi}}{1-Ce^{j\phi}}, \label{eq:moo}
%\end{equation}
%for some $C\in [0,1]$ and $\phi \in [0,2\pi]$, for each $\omega$.  
%The right side of Equation (\ref{eq:moo}) traces a region in the complex plane, as $C$  and $\phi$ are varied as indicated.  Propagation stability
%requires that the curve $k_i T_(j\omega)$ lies entirely within this region
%for each $k_i$.  To get intuition into the shape of this region, it is instructive to plot $\frac{Ce^{j\phi}}{1-Ce^{j\phi}}$  in the complex plane as a function of $\phi$, for several values of $C$ \ref{fig:region}.  From the plot and a simple algebraic analysis, it follows that the region defined by $\frac{Ce^{j\phi}}{1-Ce^{j\phi}}$, $C \in [0,1]$,
%$\phi \in [0,2\pi]$, includes all points to the right of %$\frac{-1}{2}$ in the complex plane.  Thus, the requirement for %propagation stability is that:
%\begin{equation}
%    Re(k_i T(j \omega)) \ge \frac{-1}{2}
%\end{equation}
%for all $\omega$ and each $k_i$. Equivalently, the requirement for propagation stability is that the {\em Nyquist plot} of $k_i T(j \omega)$ lies entirely to the right of $\frac{-1}{2}$ in the complex plane.   

\section{Example}

The propagation stability analysis is illustrated for an example model with planar subsystems. 
Specifically, let us consider a synchronization network process whose subsystem model $(C,A,B)$ is defined as follows: $A=\begin{bmatrix} 0 & 1 \\ 0 & d \end{bmatrix}$, $B=\begin{bmatrix} 0 \\ 1 \end{bmatrix}$, and $C=\begin{bmatrix} 1 & 0 \end{bmatrix}$, where we refer to the positive scalar $d$ as a damping constant.  Several common network models are well-approximated by this form, including the classical model for the bulk power grid's swing dynamics, mass-spring-damper network models, and models for vehicle teams engaged in formation flight.  These networks are known to exhibit a dichotomy of responses to sinusoidal or periodic disturbances, depending on the damping constant $d$.  When the system is sufficiently damped, periodic/sinusoidal disturbances cause only localized responses; on the other hand, if the damping is low, disturbances at certain frequencies incur network-wide responses.

The propagation stability concept provides a means for assessing how the damping influences the disturbance response pattern for the defined class of synhronization network processes. In particular, for this example, the frequency-domain criterion for propagation stability can readily be phrased in terms of the damping constant $d$. If the criterion is met,  then disturbance responses are necessarily localized.  If not, the network model potentially may be susceptible to network-wide responses for disturbances over certain frequency ranges,

The criterion for propagation stability includes a standard requirement for stability of the synchronization manifold, and an additional local requirement.  For the network model considered in this example, stability of the synchronization manifold has been precisely characterized in prior work.  Provided that the eigenvalues of the Laplacian are real (which encompasses the symmetric and diagonally symmetrizable cases), stability of the synchronization manifold holds for any damping.  If the Laplacian has complex eigenvalues, then a sufficiently large damping is needed for stability of the synchronization manifold.  

In our development here, we assume that the criterion for stability of the synchronization manifold is met, and focus on relating the local requirement with the damping ratio.  The local requirement is met if the subsystem model $Re(T(j\omega)) \ge -\frac{1}{2\alpha \max_i \sum_{j\ne i} g_{ij}}$ for all $\omega$. The transfer function for the planar subsystem model is $T(s)=\frac{1}{s^2+ds}$, and hence the frequency response
is $T(j\omega)=\frac{1}{-\omega^2+j d\omega}=\frac{-\omega^2-jd \omega}{\omega^4+d^2 \omega^2}$.  It immediately follows that
$Re(T(j\omega))=\frac{-1}{\omega^2+d^2}$.  Thus, we find that
$Re(T(j\omega))$ is bounded by the interval $[-\frac{1}{d^2},0]$,
with the lower bound achieved asymptotically as $\omega \rightarrow 0$.   The local requirement is therefore met if 
$-\frac{1}{d^2} \ge -\frac{1}{2\alpha \max_i \sum_{j\ne i} g_{ij}}$, or equivalently $d \ge \sqrt{2\alpha \max_i \sum_{j\ne i} g_{ij}}$.  Thus, a sufficiently large damping relative to the largest (weighted) out-degree in the network graph guarantees propagation stability.  If the damping ratio is not sufficiently large, the network is potentially susceptible to wide-area responses, for low-frequency disturbances.  

Thus, we have shown that an additional requirement of sufficient damping is needed to ensure propagation stability in addition to stability of the synchronization manifold.  In the case where the criterion is not met, disturbances with certain frequency components have the potential for amplification across the network.  On the other hand, when the damping requirement is met, disturbances are restricted to have local spheres of influence.  We note that the condition for propagation stability is phrased entirely in terms of the subsystem model and local graph properties, which then allows the development of a condition on the damping for propagation stability.

\end{document}